\documentclass{llncs}

\usepackage{times}
\usepackage{amsmath}
\usepackage{amsfonts}
\usepackage{amssymb}
\usepackage{graphicx,psfig,epsfig}
\usepackage{color}
\usepackage{algorithm}
\usepackage{algorithmic,xspace}
\usepackage{url}
\usepackage[colorlinks]{hyperref}
\usepackage{epsfig}

\newcommand{\ben}{\begin{enumerate}}
\newcommand{\een}{\end{enumerate}}

\newcommand{\hide}[1]{}

\newcommand{\field}[1]{\mathbb{#1}} 

\newcommand{\beql}[1]{\begin{equation}\label{#1}}
\newcommand{\eeq}{\end{equation}}
\newcommand{\comment}[1]{}

\newcommand{\Floor}[1]{{\left\lfloor{#1}\right\rfloor}}

\begin{document}

\title{Approximate Dynamic Programming using \\ Halfspace Queries and Multiscale Monge decomposition}

\titlerunning{Approximate Dynamic Programming} 
\author{Gary L. Miller\inst{1}, Richard Peng\inst{1}, Russell Schwartz\inst{1,2} \and Charalampos E. Tsourakakis\inst{1}}
\authorrunning{Miller et al.}   
%
\tocauthor{Richard Peng, Russell Schwartz,Charalampos Tsourakakis }
\institute{ School of Computer Science, Carnegie Mellon University, USA\\
\email{glmiller@cs.cmu.edu}, \email{yangp@cs.cmu.edu}, \email{ctsourak@cs.cmu.edu}\\ 
\and
Department of Biological Sciences, Carnegie Mellon University, USA\\
\email{russells@andrew.cmu.edu} }

\maketitle

\begin{abstract}
Let  $P=(P_1, P_2, \ldots, P_n)$, $P_i \in \field{R}$ for all $i$,
be a signal and let  $C$ be a constant.
In this work our goal is to find a function $F:[n]\rightarrow \field{R}$ which 
optimizes the following objective function: 

\begin{equation}
\min_{F} \sum_{i=1}^n (P_i-F_i)^2 + C\times |\{i:F_i \neq F_{i+1} \} |
\end{equation}

The above optimization problem reduces to solving the following recurrence, which can be done efficiently using dynamic programming in $O(n^2)$ time:
$$ OPT_i = \min_{0 \leq j \leq i-1} \left[ OPT_{j}+\sum_{k=j+1}^i \left(P_k - \frac{\sum_{m=j+1}^i P_m}{i-j} \right)^2 \right]+ C  $$ 

\hide{ 
\begin{equation*}
OPT_i = \left\{ 
\begin{array}{l r}
  \min_{0 \leq j \leq i-1} \left[ OPT_{j}+\sum_{k=j+1}^i \left(f_k - \frac{\sum_{m=j+1}^i f_m}{i-j} \right)^2 \right]+ C & \quad \text{if}~ i > 0  \\
  0 & \quad \text{if~$i=0$}     \\ 
\end{array} \right. 
\label{eq:eqrec}
\end{equation*}
}

The above recurrence arises naturally in applications where we wish to approximate the original 
signal $P$ with another signal $F$ which consists ideally of few piecewise constant segments.
Such applications include database (e.g., histogram construction), 
speech recognition, biology (e.g., denoising aCGH data) applications and many more. 

In this work we present two new techniques for optimizing dynamic programming that can handle 
cost functions not treated by other standard methods.
The basis of our first algorithm is the definition of a constant-shifted variant of the objective function
that can be efficiently approximated using state of the art methods for range searching.
Our technique approximates the optimal value of our objective function within additive
$\epsilon$ error and runs in $\tilde{O}(n^{1.5} \log{ (\frac{U}{\epsilon}) )}$ time, where $U = \max_i f_i$.
The second algorithm we provide solves a similar recurrence  within a factor of $\epsilon$ and runs in $O(n \log^2n / \epsilon )$. 
The new technique introduced by our algorithm is the decomposition of the initial problem into a small (logarithmic) number of Monge 
optimization subproblems which we can speed up using existing techniques.
\end{abstract}

\section{Introduction}
\label{sec:intro}
Dynamic programming is a widely used problem solving technique with applications in various fields
such as operations research, biology, speech recognition, time series analysis and many others. 
Due to its importance a lot of work has focused on speeding up naive dynamic programming
implementations. Such techniques include the Knuth-Yao technique\cite{DBLP:journals/acta/Knuth71,fyao,804691}, a special case of the use
of totally monotone matrices \cite{1109562}, the SMAWK algorithm for finding the row-minima of totally monotone
matrices\cite{smawk}, and several other techniques exploiting special properties such as the convexity and 
concavity \cite{Eppstein88speedingup,135865}. The basis of these techniques lies in the theory of Monge properties for optimization \cite{240860}, 
which date back to the French mathematician Monge \cite{monge}. 

In this work, we consider the following recurrence, where $P \in \field{R}^n$ and $C$ is a constant:

$$ OPT_0 = 0 \text{,~~}OPT_i = \min_{0 \leq j \leq i-1} \left[ OPT_{j}+\sum_{k=j+1}^i \left(P_k - \frac{\sum_{m=j+1}^i P_m}{i-j} \right)^2 \right]+ C \text{, for~} i>0  $$ 

This recurrence arises naturally in several applications where one wants to approximate a given signal $f$ with a signal $F$
which ideally consists of few piecewise constant segments. Such applications include 
histogram construction in databases (e.g., \cite{671191,1132873,543637}), determining DNA
copy numbers in cancer cells from micro-array data (e.g., \cite{picard,trimmer}), speech recognition, 
data mining (e.g., \cite{DBLP:conf/sdm/TerziT06}) and many others.

In this work we present two new techniques for optimizing dynamic programming that can handle 
cost functions not treated by other standard methods.
The basis of our first algorithm is the definition of a constant-shifted variant of the objective function
that can be efficiently approximated using state of the art methods for range searching.
Our technique approximates the optimal value of our objective function within additive
$\epsilon$ error and runs in $\tilde{O}(n^{1.5} \log{ (\frac{U}{\epsilon}) )}$ time, where $U = \max_i f_i$.
The second algorithm we provide solves a similar recurrence  within a factor of $\epsilon$ and runs in $O(n \log^2n / \epsilon )$. 
The new technique introduced by our algorithm is the decomposition of the initial problem into a small (logarithmic) number of Monge 
optimization subproblems which we can speed up using existing techniques.

The remainder of the paper is organized as follows:
Section~\ref{sec:prelim} presents briefly existing work on the problem 
and the necessary background.  Section~\ref{sec:method} presents our proposed algorithms
and their theoretical analysis and Section~\ref{sec:concl} concludes the
paper with a brief summary and discussion.

\section{Background}
\label{sec:prelim}
In this section, we briefly summarize existing work on speeding up dynamic programming. 
First, we briefly present existing work on speeding up dynamic programming.
Then, we present state of the art results on reporting points in halfspaces, an important
case of range queries in computational geometry, and on optimizing Monge functions.

\subsection{Speeding up Dynamic Programming} 

Dynamic programming is a powerful problem solving technique introduced by Bellman \cite{862270}
with numerous applications in biology (e.g., \cite{picard,360861,9303}), in control theory (e.g., \cite{517430}), 
in operations research and many other fields. 
Due to its importance, a lot of research has focused on speeding up naive dynamic programming implementations.
A successful example of speeding up a naive dynamic programming implementation is the computation of optimal binary search trees.
Gilbert and Moore solved the problem efficiently using dynamic programming \cite{bb12353}.
Their algorithm runs in $O(n^3)$  time and for several years this running time was considered to be tight. 
In 1971 Knuth \cite{DBLP:journals/acta/Knuth71} showed that the same computation can be carried out in $O(n^2)$ time. 
This remarkable result was generalized by Frances Yao in \cite{fyao,804691}.
Specifically, Yao showed that this dynamic programming speedup technique works for a large class
of recurrences. She considered the recurrence $c(i,i) =0$, $c(i,j) = \min_{i<k\leq j}{ (  c(i,k-1) + c(k,j)) }  + w(i,j)$
for $i <j$ where the weight function $w$ satisfies the quadrangle inequality (see Section~\ref{sec:monge}) and proved
that the solution of this recurrence can be found in $O(n^2)$ time. 
Eppstein, Galil and Giancarlo have considered similar recurrences where they showed that naive $O(n^2)$ implementations
of dynamic programming can run in $O(n\log{n})$ time \cite{Eppstein88speedingup}. Furthermore, 
Eppstein, Galil, Giancarlo and Italiano have also explored the effect of sparsity \cite{146650,146656}, 
another key concept in speeding up dynamic programming.  
Aggarwal, Klawe, Moran, Shor, Wilber developed an algorithm, widely known as the SMAWK algorithm,
\cite{10546} which can compute in $O(n)$ time the row maxima of a totally monotone $n \times n$ matrix. 
The connection between the Knuth-Yao technique and the SMAWK algorithm was made 
clear in \cite{1109562}, by showing that the Knuth-Yao technique is a special case of the use of totally monotone matrices. 
The basic properties which allow these speedups are the convexity or concavity of the weight function.
The study of such properties data back to Monge \cite{monge} and are well studied in the literature, see for example \cite{240860}. 

Close to our work, lies the work on histogram construction, an important problem for database applications. 
Jagadish et al. \cite{671191} originally provided a simple dynamic programming 
algorithm which runs in $O(kn^2)$ time, where $k$ is the number of buckets and $n$ the input size
and outputs the best V-optimal histogram. 
Guha, Koudas and Kim \cite{1132873,543637} propose a $(1+\epsilon)$  approximation algorithm which runs in linear time.
Their algorithms exploits monotonicity properties of the key quantities involved in the problem.

\subsection{Reporting Points in a Halfspace} 

Let $S$ be a set of points in $d$-dimensional space $\field{R}^d$.
Consider the problem of preprocessing $S$ such that for any 
halfspace query $\gamma$ we can report the set $S \cap \gamma$ 
efficiently. This problem is a special case of range searching. For
an extensive survey see \cite{Agarwal99geometricrange}. 
For $d=2$, the problem has been solved optimally by \cite{chazelle_power_1985}.  For $d=3$, \cite{323248} gave a solution with nearly linear space and $O(\log{n}+k)$
query time, while \cite{100260} gave a solution with a more expensive preprocessing 
but $O(n\log{n})$ space. 
For dimensions $d \geq 4$, \cite{82363} gave an algorithm that requires with $O(n^{\Floor{d/2} +\epsilon})$ preprocessing time and space,
where $\epsilon$ is an arbitrarily small positive constant, and can subsequently
answer queries in $O(\log{n}+k)$ time.  
Matousek in \cite{DBLP:conf/focs/Matousek91} provides improved results on the problem,
which are conjectured to be optimal up to $O(n^{\epsilon})$ or polylogarithmic factors.  
Now, we refer to the main theorem of \cite{DBLP:conf/focs/Matousek91} on the emptiness decision problem,
i.e., determining whether $S \cap \gamma$  empty or not, phrased as theorem 1.3:

\begin{theorem}\cite{DBLP:conf/focs/Matousek91} 
Given a set of $S$ of $n$ points in $\field{R}^d$, $d \geq 4$, one can build 
in time $O(n^{1+\delta})$  a linear size data structure 
which can decide whether a query halfspace contains a point of $S$ 
in time $O(n^{1-1/\Floor{\frac{d}{2} } } 2^{c'\log^*{n}})$, where $\delta > 0$
is an arbitrarily small constant and $c'=c'(d)$ is a constant depending on the dimension.
\end{theorem}

\subsection{Monge Functions and Dynamic Programming}
\label{sec:monge}
Here, we refer to a theorem in \cite{Eppstein88speedingup} which we use in Section~\ref{subsec:3point2}.
A function w defined on pairs of integer indices is Monge if for any 4-tuple of indices $i_1 < i_2 < i_3 < i_4$,
$w(i_1,i_4)+w(i_2,i_3) \geq w(i_1, i_3) + w(i_2, i_4)$. Then the following Theorem holds: 

\hide{This can be viewed as a discrete definition of
convexity and is the assumption used in various works on optimizing dynamic programming \cite{fyao,dknuth}
and is a direct result of total monotonicity \cite{smawk,1644032}.
One such theorem is stated as follows as phrased in \cite{Eppstein88speedingup}: }

\begin{theorem}[\cite{Eppstein88speedingup}]

Given a Monge function $w: \{0 \dots n\} \times  \{0 \dots n\}  \rightarrow \field{R}$, and a vector $(a_0, a_1 \dots a_{n-1})$
the value of $\min_{j < i} \{ a_j + w(j,i)\}$ can be calculated for each $i=1,\ldots,n$ given
the previous values $a_0, a_1, \dots a_{i-1}$ in  $O(n\log{n})$ time total.

\end{theorem}

\section{Proposed Method}
\label{sec:method}
In the following, let the initial vector be $(P_1, \dots, P_n)$ and let $S_i = \sum_{j=1}^i P_j$. 
The transition function for the dynamic programming for $i>0$ can be rewritten as:

\begin{eqnarray} 
OPT_i &= \min_{j<i} OPT_j + \min_x \sum_{m=j+1}^i(P_m-x)^2 + C \nonumber \\
&= \min_{j<i} OPT_j + \sum_{m=j+1}^i P_m^2 - \frac{(S_i- S_j)^2}{i-j} + C 
\label{eq:eqopt}
\end{eqnarray}

The transition can be viewed as a weight function $w(j,i)$ that takes the two indices $j$ and
$i$ as parameters such that $w(j,i) =  \sum_{m=j+1}^i P_m^2 - \frac{(S_i- S_j)^2}{i-j} + C$.
The dynamic programming is then equivalent to a shortest path from $0$ to $n$.

The weight function does not have the Monge property, as demonstrated by the vector
$P_0 = 1, P_2 = 2, P_3 = 0, P_4 = 2, \ldots, P_{2k-1} = 0, P_{2k} = 2, P_{2k+1} = 1$.
When $C = 1$, thee optimal choices of $j$ for $i  = 1, \dots, 2k$ are $j = i-1$, i.e., we fit one segment per point. 
However, once we add in $P_{2k+1}$ the optimal solution changes to to fitting all points on a single segment. 
Therefore, choosing a transition to $j_1$ over one to $j_2$ at some $i$ does not allow us to discard $j_2$ from future considerations.
This is one of the main difficulties for applying techniques based on the increasing order of decision points,
such as the method of Eppstein {\em et al.} \cite{Eppstein88speedingup}, to reduce the complexity of the $O(n^2)$
algorithm in \cite{trimmer}. 

Let $DP_i = OPT_i - \sum_{m=1}^i P_m^2$. 
We claim that $DP_i$ is the solution to a simpler optimization problem:

\begin{lemma} 
$DP_i$ satisfies the following optimization formula:

\begin{equation}
DP_i = \left\{ 
\begin{array}{l r}
  \min_{j<i} DP_j - \frac{(S_i- S_j)^2}{i-j} + C & \text{if}~ i > 0  \\
  0                                              & \text{if}~i=0     \\ 
\end{array} \right. 
\label{eq:eqdp}
\end{equation}

\end{lemma} 

\begin{proof}

As $\sum_{m=1}^0 P_m^2=0$, the result is true for $i=0$.
If $i>0$, substituting in equation~\ref{eq:eqopt} gives:

\begin{align*}
DP_i &= OPT_i - \sum_{m=1}^i P_m^2\\
&= \min_{j<i} OPT_j  - \frac{(S_i- S_j)^2}{i-j} + \sum_{m=j+1}^i P_m^2 - \sum_{m=1}^i P_m^2 + C\\
&= \min_{j<i} DP_j + \sum_{m=1}^j P_m^2 - \frac{(S_i- S_j)^2}{i-j} + \sum_{m=j+1}^i P_m^2 - \sum_{m=1}^i P_m^2 + C\\
&= \min_{j<i} DP_j - \frac{(S_i- S_j)^2}{i-j} + C
\end{align*}

Elimination of the terms $ \sum_{m=1}^i P_m^2 $ from both sides gives our result. 
Observe that the second order moments of $OPT_i$ are absent from $DP_i$. 

We can use $\tilde{w}(j,i)$ to denote the shifted weight function, aka. $\tilde{w}(j,i) = - \frac{(S_i- S_j)^2}{i-j} + C$, it's easy to check that $w(j,i) = \tilde{w}(j,i) +  \sum_{m=i}^j P_m^2 $

\end{proof}

\subsection{$\tilde{O}(n^{1.5} \log{ (\frac{U}{ \epsilon }  ) } )$ algorithm to approximate within additive $\epsilon$}

\begin{algorithm}[!htb]
\caption{Approximation within additive $\epsilon$ using 4D halfspace queries}
\begin{algorithmic}
\STATE initialize 4D halfspace query structure $Q$
\FOR{$i=1$ to $n$}
	\STATE $low \gets 0$
	\STATE $high \gets nU^2$
	\WHILE{$high - low > \epsilon/n$}
		\STATE $m \gets (low+high)/2$
		\IF{ $Q.intersect((m, i, S_i, -1)^T x \leq m\cdot i + S_i^2)$}
			\STATE $high \gets m$
		\ELSE
			\STATE $low \gets m$
		\ENDIF
	\ENDWHILE
	\STATE $\tilde{DP}_i \gets (low+high)/2$
	\STATE $x \gets (i, \tilde{DP}_i, 2S_i, S_i^2+\tilde{DP}_ii)$
	\STATE $Q.insert(x)$
\ENDFOR
\end{algorithmic}\end{algorithm}

Our proposed method (as shown in Algorithm 1) uses the results of \cite{DBLP:conf/focs/Matousek91} to obtain a fast
algorithm for the additive approximation variant of the problem. \hide{Specifically, the algorithm maintains a 4-dimensional halfspace query data structure and tracks potential
transitions by inserting points corresponding to the DP states into them.} Specifically, the algorithm initializes a 4-dimensional halfspace query data structure. 
The algorithm then uses binary searches to compute an accurate estimate of the value $DP_i$ for $i=1,\ldots,n$.
As errors are introduced at each term, we use $\tilde{DP_i}$ to denote the approximate value of $DP_i$ calculated
by earlier iterations of the binary search, and $\bar{DP_i}$ to be the optimum value of the transition function computed
using the approximate values of $\tilde{DP_j}$. Formally:

$\bar{DP}_i = \min_{j<i} \tilde{DP}_j - \underbrace{\frac{(S_i- S_j)^2}{i-j}}_{\tilde{w}(j,i)} + C$. 

Theorem 3 shows that it's sufficient to approximate $\tilde{DP}_i$ to within an additive $\epsilon/n$ of $\bar{DP}_i$
in order to approximate $DP_n$ within additive $\epsilon$. 
Let $U = \max{ \{ \sqrt{C}, P_1, \dots, P_n\} } $,
the maximum value of the objective function is upper-bounded by the cost one would incur from declaring there is single interval with $x=0$, giving a bound of $U^2n$.
Therefore $O(\log(\frac{U^2n}{\epsilon/n})) = \tilde{O}( \log{ (\frac{U}{ \epsilon } ) } )$ iterations of binary search
at each $i$ are sufficient. 

To check whether $\bar{DP}_i \geq x+C$, we need to solve the decision problem of whether there
exists a $j < i$ such that the following inequality holds:

\begin{align*}
x + C &\geq \tilde{DP}_i \\
\exists{j<i}, x &\geq \tilde{DP_j} - \frac{(S_i- S_j)^2}{i-j}\\
\exists{j<i}, x(i-j) &\geq \tilde{DP_j}(i-j) - (S_i-S_j)^2\\
\exists{j<i}, xi + S_i^2 &\geq xj + \tilde{DP_j}i + 2S_iS_j - S_j^2 - \tilde{DP_j}j
\end{align*}

The term on the right hand side can be interpreted as the dot product between $(x, i, S_i, -1)$
and $(j, \tilde{DP_j}, 2S_j, S_j^2 + j\tilde{DP_j})$. If we think of the values $(j, \tilde{DP_j}, 2S_j, S_j^2+\tilde{DP_j}j)$ as points in $\field{R}^4$,
the decision problem becomes whether the intersection of a point set with a halfplane is null. If the
point set has size $n$, this can be done in $\tilde{O}(n^{0.5})$ per query and $O(\log{n})$ amortized
time per insertion of a point\cite{DBLP:conf/focs/Matousek91}.
So the optimal value of $DP_i$ can be found within an additive
constant of $\epsilon/n$ using the binary search in $\tilde{O}(n^{0.5}\log{ (\frac{U}{ \epsilon } ) } )$ time. Therefore,
we can proceed along the indices from 1 to $n$, find the approximately optimal value of $OPT_i$ and insert a point
corresponding to it into the query structure, getting an algorithm that runs in $\tilde{O}(n^{1.5} \log{ (\frac{U}{ \epsilon } ) } )$ time.

The following theorem states that a small error at each step suffices to give an overall good approximation.
We show inductively that if $\tilde{DP}_i$ approximates $\bar{DP}_i$ within $\epsilon/n$, $\tilde{DP}_i$ is within $i\epsilon /n$ additive error from the optimal value $DP_i$.
For the proof of Theorem~\ref{thrm:mainthrm}, see the Appendix~\ref{sec:appendix}.

\begin{theorem}
\label{thrm:mainthrm}
Let $\tilde{DP}_i$ be the approximation of our algorithm to $DP_i$. Then, the following inequality holds: 
\begin{equation} 
|DP_i - \tilde{DP}_i | \leq \frac{\epsilon i}{n}
\end{equation} 

\end{theorem}

By substituting in Theorem 3 $i=n$ we obtain the following Corollary, proving that our algorithm 
is an approximation algorithm within $\epsilon$ additive error. 

\begin{corollary}
Let $\tilde{DP}_n$ be the approximation of our algorithm to $DP_n$. Then, the following inequality holds:
\begin{equation} 
|DP_n - \tilde{DP}_n | \leq \epsilon
\end{equation} 
\end{corollary}

\hide{
\subsection{$O(n \log^3n / \epsilon )$ algorithm to approximate the shifted objective within a factor of $\epsilon$}

We show an algorithm that approximates the changed objective function within a factor of $\epsilon$.
Since the term $\sum_{i=1}^n P_i^2$ can dominate the result, this does not imply a good approximation of
$OPT$.  In the aCGH denoising application, however, one would expect this to produce a useful segmentation. 
}

\subsection{$O(n \log^2n / \epsilon )$ algorithm to approximate within multiplicative $\epsilon$}
\label{subsec:3point2}

Once again, consider the transition function $w$ in the optimization
formula for $OPT_i$. Our approach is based on approximating $w$
with a logarithmic number of Monge functions, while incurring a multiplicative
error of at most $\epsilon$. When viewed from the context of a shortest path problem, we are perturbing
the weight of each edge by $\epsilon$. So as long as the weight of each
edge is positive, the length of any path, and therefore the optimal answer computed
in the perturbed version, is off by a factor of $\epsilon$ as well.

The main idea of our algorithm is as follows:
we break our initial problem into a small (logarithmic) number of Monge optimization
subproblems which we can speed up using existing techniques, e.g., \cite{Eppstein88speedingup}.
We achieve this by detecting which part of $w(j,i)$ causes $w(j,i)$ not to be Monge,
and then by finding intervals in which we can approximate it accurately with a constant. 
We also make sure the optimal breakpoints of the Monge sub-problems lie in the
specified subintervals by making the function outside that subinterval arbitrarily large while
maintaining its Monge property.

\begin{algorithm}[!htb]
\caption{Approximation within factor of $\epsilon$ using Monge function search}
\begin{algorithmic}
\STATE Maintain $m = \log{n}/\log{(1+\epsilon)}$ Monge function search data structures $Q_1,\dots,Q_m$
\COMMENT{Each $Q_k$ corresponds to a Monge function $w_k(j,i)$ such that
$w_k(j,i) = (\sum_{m = j+1}^i (i-j)P_m^2 - (S_i-S_j)^2)/(1+\epsilon)^k) + C$ if
$(1+\epsilon)^{k-1} \leq i-j \leq (1+\epsilon)^k$, otherwise  arbitrarily large .}

\STATE $OPT_0 \gets 0$
\FOR {$k = 1$ to $m$}
\STATE $Q_k.a_0 \gets 0$
\ENDFOR

\FOR {$i = 1$ to $n$}
	\STATE $OPT_i \gets \infty$
	\FOR {$k = 1$ to $m$}
		\STATE $\text{localmin}_k \gets \min_{j<i} Q_k.a_j + w_k(j,i)$.
		\STATE $OPT_i \gets \min\{OPT_i, \text{localmin} + C \}$		
	\ENDFOR
	\FOR {$k = 1$ to $m$}
		\STATE $Q_k.a_i \gets OPT_i$
	\ENDFOR	
\ENDFOR
\end{algorithmic}
\end{algorithm}

We let $w'(j,i) = \sum_{m = j+1}^i (i-j)P_m^2 - (S_i-S_j)^2$. In other
words, $w(j,i) = w'(j,i)/(i-j) + C$. Since $Var(X)=E(X^2)-E(X)^2$,
an alternate formulation of $w'(j,i)$ is:

$$w'(j,i) = \sum_{j+1 \leq m_1 < m_2 \leq i} (P_{m_1} - P_{m_2})^2$$

\begin{lemma}
w'(j,i) is Monge, in other words, for any $i_1<i_2<i_3<i_4$, $w'(i_1, i_4)+w'(i_2,i_3) \geq w'(i_1,i_3) + w'(i_2, i_4)$.
\end{lemma}

\begin{proof}
Since each term in the summation is non-negative, it suffices to show that any pair of $(m_1,m_2)$ is summed as many
times on the left hand side as on the right hand side. If $i_2+1 \leq m_1 < m_2 \leq i_3$, each term is counted
twice on each side. Otherwise, each term is counted  once on the left hand side since $i_1+1 \leq m_1 < m_2 \leq i_4$
and at most once on the right hand side since $[i_1+1, i_3] \cap [i_2+1, i_4] = [i_2+1, i_3]$.
\end{proof}

Also, as $w'(i,j)$, $i-j$ and $C$ are all non-negative, approximating $w'(i,j)/(i-j)$
within a factor of $\epsilon$
gives an approximation of $w(i,j)$ within a factor of $\epsilon$.
For each $i$, we try to approximate $i-j$ ($j<i$) with a constant $c'$ such that $c' \leq i-j \leq c'(1+\epsilon)$.
Since $1 \leq i-j \leq n$, we only need $\log_{(1+\epsilon)}n  = \log{n} / \log (1+\epsilon) $ distinct values
of $c'$ for all transitions. Note that this is equivalent to $j$ being in the range$[l,r] = [ i-c'(1+\epsilon), i-c']$.

Since $c'$ is a constant, $w'(j,i)/c'$ is also a Monge function. However, notice that we can only use $j$ when
$c' \leq i-j \leq c'(1+\epsilon)$. This constraint on pairs $(j, i)$ can be enforced by
setting $w_k(j,i)$ to arbitrarily large values when $(j,i)$ do not satisfy the condition.
This ensures that $j$ will not be used as a breakpoint for $i$.
Furthermore, $w_k$ needs to be adjusted to remain Monge in order to keep the assumptions of Theorem 2 valid.
Since the $[i_1+1, i_4]$ is the longest and $[i_2+1, i_3]$ is the shortest range respectively, one possibility is to assign
exponentially large values to very long and short ranges. The following is one possibility
when $M$ is an arbitrarily large positive constant:

\begin{equation}
w_k(j,i) = \left\{ 
\begin{array}{l r}
  2^{n-i+j}M & i - j < c' = (1+\epsilon)^{k-1} \\
  2^{i-j}M & i - j > (1+\epsilon)c' = (1+\epsilon)^{k} \\
  w'(j,i)/c'           & \text{otherwise}     \\ 
\end{array} \right. 
\label{eq:eqwk}
\end{equation}

\begin{lemma}
$w_k$ is Monge. That is, for any 4-tuple $i_1 < i_2 < i_3 < i_4$, $w_k(i_1, i_4)+w_k(i_2,i_3) \geq w_k(i_1,i_3) + w_k(i_2, i_4)$.
\end{lemma}

\begin{proof}
As $M$ is arbitrarily large, we may assume $w_k(j,i) \geq w'(j,i)$. If $c' \leq i_3-i_1, i_4-i_2 \leq (1+\epsilon)c'$, then
$w_k(i_1,i_3) + w_k(i_2, i_4) = (1/c')(w'(i_1,i_3) + w'(i_2, i_4)) \leq (1/c')(w'(i_1,i_4) + w'(i_2, i_3)) \leq w_k(i_1,i_4) + w_k(i_2, i_3)$

We assume $i_3-i_1 \leq i_4-i_2$ since the mirror case can be considered similarly.
Suppose $i_3-i_1 < c'$ and $i_4-i_2 \leq c'(1+\epsilon)$.
Then as $i_2 > i_1$, $i_3-i_2 \leq i_3-i_1 - 1$. $w_k(i_2, i_3)  = 2^{n-i_3+i_2}M \geq 2 \cdot 2^{n-i_3-i_1} \geq w_k(i_1, i_3) + w_k(i_2, i_4)$
The case of $c' \leq i_3-i_1$ and $c'(1+\epsilon) < i_4-i_2$ can be done similarly.

Suppose $i_3-i_1 < c'$ and $c'(1+\epsilon) < i_4-i_2$. Then $i_3 - i_2 < i_3-i_1$ and $i_4 - i_1 > i_4-i_2$ gives
$w_k(i_1, i_4) \geq w_k(i_2, i_4)$ and $w_k(i_2, i_3) \geq w_k(i_1, i_3)$. Adding them gives the desired result.

\end{proof}

\hide{Bug corrected $[l,r] = [(1-c'(1+\epsilon))i , (1-c')i]$}
So the equation for $OPT_i$ becomes: 

\begin{align*}
OPT_i &= \min_{j \in [l,r]} OPT_j + \frac{w'(i,j)}{i-j} + C\\
&\approx \min_{j \in [l,r]} OPT_j + w_k(j,i) + C
\end{align*}

Note that storing values of the form $2^kM$ using only their exponent $k$ suffices for comparison,
so this adjustment does result in any change in runtime.
By Theorem 2, the Monge function optimization problem for each
$w_k$ can be solved in $O(n\log{n})$ time, giving a total runtime of $O(n\log^2{n}/\epsilon)$.
Pseudocode for this algorithm is shown in Algorithm 2. The algorithm uses $m = \log{n} / \epsilon$
copies of the algorithm mentioned in theorem 2 implicitly.

\section{Conclusions}
\label{sec:concl}
In this work we present two new techniques for optimizing dynamic programming that can handle 
cost functions not treated by other standard methods. The first algorithm 
approximates the optimal value of our objective function within additive
$\epsilon$ error and runs in $\tilde{O}(n^{1.5} \log{ (\frac{U}{\epsilon}) )}$ time, where $U = \max_i f_i$. 
The efficiency of our algorithm is based on the work of Jir\'{\i} Matousek \cite{DBLP:conf/focs/Matousek91},
since our binary search on the value of a costant-shifted variant of the objected reduces 
to performing halfspace queries. 
The second algorithm solves a similar recurrence within a factor of $\epsilon$ and runs in $O(n \log^2n / \epsilon )$. 
The new technique introduced by our algorithm is the decomposition of the initial problem into a small (logarithmic) number of Monge 
optimization subproblems which we can speed up using existing techniques \cite{Eppstein88speedingup}. 

While the recurrences we solve are not treated using existing techniques for speeding up dynamic 
programming in an exact way, the results we obtain suggest that the $O(n^2)$ bound 
is not tight, i.e., there exists more structure which we can take advantage of. For example, the following lemma holds(see Appendix for a proof):

\begin{lemma}
If $|P_{i_1}-P_{i_2}|>2\sqrt{2C}$, then in the optimal solution of the dynamic programming using $L_2$ norm, $i_1$ and $i_2$ are
in different segments.
\label{thrm:lemmalast}
\end{lemma}

In future work, we plan to exploit the structure inherent to our problem to obtain a faster, exact algorithm. 


\hide{
\section{Acknowledgments}
R.S. and C.E.T. were supported in this work by U.S. National
Institute of Health award 1R01CA140214.
R.P. was supported in this work by Natural Sciences and Engineering
Research Council of Canada (NSERC),  under Grant PGS M-377343-2009.
G.L.M. was supported in this work by the National  Science Foundation under Grant No. CCF-0635257.

Any opinions, findings, and conclusions or recommendations expressed in this material are those of
the author(s) and do not necessarily reflect the views of the National Institute of Health and National Science Foundation, or other funding parties.
}

\bibliographystyle{plain}
\bibliography{paper}

\begin{thebibliography}{10}

\bibitem{Agarwal99geometricrange}
Pankaj~K. Agarwal and Jeff Erickson.
\newblock Geometric range searching and its relatives.
\newblock In {\em Advances in Discrete and Computational Geometry}, pages
  1--56. American Mathematical Society, 1999.

\bibitem{100260}
A.~Aggarwal, M.~Hansen, and T.~Leighton.
\newblock Solving query-retrieval problems by compacting voronoi diagrams.
\newblock In {\em STOC '90: Proceedings of the twenty-second annual ACM
  symposium on Theory of computing}, pages 331--340, New York, NY, USA, 1990.
  ACM.

\bibitem{10546}
A~Aggarwal, M~Klawe, S~Moran, P~Shor, and R~Wilber.
\newblock Geometric applications of a matrix searching algorithm.
\newblock In {\em SCG '86: Proceedings of the second annual symposium on
  Computational geometry}, pages 285--292, New York, NY, USA, 1986. ACM.

\bibitem{smawk}
Alok Aggarwal, Maria~M. Klawe, Shlomo Moran, Peter~W. Shor, and Robert~E.
  Wilber.
\newblock Geometric applications of a matrix-searching algorithm.
\newblock {\em Algorithmica}, 2:195--208, 1987.

\bibitem{1109562}
Wolfgang~W. Bein, Mordecai~J. Golin, Lawrence~L. Larmore, and Yan Zhang.
\newblock The knuth-yao quadrangle-inequality speedup is a consequence of
  total-monotonicity.
\newblock In {\em SODA '06: Proceedings of the seventeenth annual ACM-SIAM
  symposium on Discrete algorithm}, pages 31--40, New York, NY, USA, 2006. ACM.

\bibitem{862270}
Richard~Ernest Bellman.
\newblock {\em Dynamic Programming}.
\newblock Dover Publications, Incorporated, 2003.

\bibitem{517430}
Dimitri~P. Bertsekas.
\newblock {\em Dynamic Programming and Optimal Control}.
\newblock Athena Scientific, 2000.

\bibitem{240860}
Rainer~E. Burkard, Bettina Klinz, and R\"{u}diger Rudolf.
\newblock Perspectives of monge properties in optimization.
\newblock {\em Discrete Appl. Math.}, 70(2):95--161, 1996.

\bibitem{chazelle_power_1985}
Bernard Chazelle, Leonidas~J. Guibas, and {Der-Tsai} Lee.
\newblock The power of geometric duality.
\newblock {\em {BIT}}, 25(1):76―90, 1985.

\bibitem{323248}
Bernard Chazelle and Franco~P. Preparata.
\newblock Halfspace range search: an algorithmic application of k-sets.
\newblock In {\em SCG '85: Proceedings of the first annual symposium on
  Computational geometry}, pages 107--115, New York, NY, USA, 1985. ACM.

\bibitem{82363}
K.~L. Clarkson and P.~W. Shor.
\newblock Applications of random sampling in computational geometry, ii.
\newblock {\em Discrete Comput. Geom.}, 4(5):387--421, 1989.

\bibitem{Eppstein88speedingup}
David Eppstein, Zvi Galil, and Raffaele Giancarlo.
\newblock Speeding up dynamic programming.
\newblock In {\em In Proc. 29th Symp. Foundations of Computer Science}, pages
  488--496, 1988.

\bibitem{146650}
David Eppstein, Zvi Galil, Raffaele Giancarlo, and Giuseppe~F. Italiano.
\newblock Sparse dynamic programming i: linear cost functions.
\newblock {\em J. ACM}, 39(3):519--545, 1992.

\bibitem{146656}
David Eppstein, Zvi Galil, Raffaele Giancarlo, and Giuseppe~F. Italiano.
\newblock Sparse dynamic programming ii: convex and concave cost functions.
\newblock {\em J. ACM}, 39(3):546--567, 1992.

\bibitem{135865}
Zvi Galil and Kunsoo Park.
\newblock Dynamic programming with convexity, concavity and sparsity.
\newblock {\em Theor. Comput. Sci.}, 92(1):49--76, 1992.

\bibitem{bb12353}
E.N. Gilbert and E.F. Moore.
\newblock Variable-length binary encodings.
\newblock {\em Bell System Tech.}, 38:933--966, July 1959.

\bibitem{1132873}
Sudipto Guha, Nick Koudas, and Kyuseok Shim.
\newblock Approximation and streaming algorithms for histogram construction
  problems.
\newblock {\em ACM Trans. Database Syst.}, 31(1):396--438, 2006.

\bibitem{543637}
Sudipto Guha, Nick Koudas, and Divesh Srivastava.
\newblock Fast algorithms for hierarchical range histogram construction.
\newblock In {\em PODS '02: Proceedings of the twenty-first ACM
  SIGMOD-SIGACT-SIGART symposium on Principles of database systems}, pages
  180--187, New York, NY, USA, 2002. ACM.

\bibitem{360861}
D.~S. Hirschberg.
\newblock A linear space algorithm for computing maximal common subsequences.
\newblock {\em Commun. ACM}, 18(6):341--343, 1975.

\bibitem{671191}
H.~V. Jagadish, Nick Koudas, S.~Muthukrishnan, Viswanath Poosala, Kenneth~C.
  Sevcik, and Torsten Suel.
\newblock Optimal histograms with quality guarantees.
\newblock In {\em VLDB '98: Proceedings of the 24rd International Conference on
  Very Large Data Bases}, pages 275--286, San Francisco, CA, USA, 1998. Morgan
  Kaufmann Publishers Inc.

\bibitem{DBLP:journals/acta/Knuth71}
Donald~E. Knuth.
\newblock Optimum binary search trees.
\newblock {\em Acta Inf.}, 1:14--25, 1971.

\bibitem{DBLP:conf/focs/Matousek91}
Jir\'{\i} Matousek.
\newblock Reporting points in halfspaces.
\newblock In {\em FOCS}, pages 207--215, 1991.

\bibitem{monge}
Gaspard Monge.
\newblock Memoire sue la theorie des deblais et de remblais.
\newblock {\em Histoire de l’Academie Royale des Sciences de Paris, avec les
  Memoires de Mathematique et de Physique pour la meme annee}, pages 666--704,
  1781.

\bibitem{picard}
F.~Picard, S.~Robin, M.~Lavielle, C.~Vaisse, and J.~J. Daudin.
\newblock A statistical approach for array cgh data analysis.
\newblock {\em BMC Bioinformatics}, 6, 2005.

\bibitem{DBLP:conf/sdm/TerziT06}
Evimaria Terzi and Panayiotis Tsaparas.
\newblock Efficient algorithms for sequence segmentation.
\newblock In {\em SDM}, 2006.

\bibitem{trimmer}
C.E. Tsourakakis, D.~Tolliver, Maria~A. Tsiarli, S.~Shackney, and R.~Schwartz.
\newblock Cghtrimmer: Discretizing noisy array cgh data.
\newblock {\em submitted, available at Arxiv: http://arxiv.org/abs/1002.4438},
  2010.

\bibitem{9303}
Michael~S Waterman and Temple~F Smith.
\newblock Rapid dynamic programming algorithms for rna secondary structure.
\newblock {\em Adv. Appl. Math.}, 7(4):455--464, 1986.

\bibitem{fyao}
F.~Yao.
\newblock Speed-up in dynamic programming.
\newblock {\em SIAM J. Alg. Disc. Methods}, 3:532--540, 1982.

\bibitem{804691}
F.~Frances Yao.
\newblock Efficient dynamic programming using quadrangle inequalities.
\newblock In {\em STOC '80: Proceedings of the twelfth annual ACM symposium on
  Theory of computing}, pages 429--435, New York, NY, USA, 1980. ACM.

\end{thebibliography}

\section*{Appendix} 
\label{sec:appendix}
The proof of Theorem~\ref{thrm:mainthrm} follows:

\begin{proof}

We use induction on the number of points. 
Using the same notation as above, let $\bar{DP}_i = \min_{j<i} \tilde{DP}_j - w(j,i) + C$. 
By construction the following inequality holds: 

\begin{equation}
|\bar{DP}_i - \tilde{DP}_i | \leq \frac{\epsilon}{n} ~\forall i=1,\dots,n
\label{eq:bartilde}
\end{equation}

When $i=1$ it is clear that $|DP_1 - \tilde{DP}_1| \leq  \frac{\epsilon}{n}$.
Our inductive hypothesis is the following: 

\begin{equation}
|DP_j - \tilde{DP}_j | \leq \frac{j\epsilon}{n} ~\forall j<i
\label{eq:inductivehypothesis}
\end{equation} 

It suffices to show that the following inequality holds: 

\begin{equation} 
|DP_i - \bar{DP}_i | \leq \frac{(i-1)\epsilon}{n} 
\label{eq:suffices}
\end{equation}
since then by the triangular inequality we obtain: 
$$\frac{ i\epsilon }{n} \geq |DP_i - \bar{DP}_i |+ |\bar{DP}_i - \tilde{DP}_i | \geq |DP_i - \tilde{DP}_i |.$$ 

Let $j^*,\bar{j}$ be the optimum breakpoints for $DP_i$ and $\bar{DP}_i$ respectively,  $j^*,\bar{j} \leq i-1$.

\begin{align*}
DP_i &= DP_{j^*}+ \tilde{w}(j^*,i)+C \leq \\
& \leq DP_{\bar{j}}+ \tilde{w}(\bar{j},i)+C \leq \\ 
& \leq \tilde{DP}_{\bar{j}} + \tilde{w}(\bar{j},i)+C+ \frac{\bar{j}\epsilon}{n} \text{(by \ref{eq:inductivehypothesis})}= \\
& =\bar{DP}_i +  \frac{\bar{j}\epsilon}{n} \leq \\ 
& \leq \bar{DP}_i +  \frac{(i-1)\epsilon}{n}
\end{align*}

Similarly we obtain:

\begin{align*}
\bar{DP}_i &= \tilde{DP}_{\bar{j}}+ \tilde{w}(\bar{j},i)+C \leq \\
& \leq \tilde{DP}_{j^*}+ \tilde{w}(j^*,i)+C \leq \\ 
& \leq  DP_{j^*}+ \tilde{w}(j^*,i)+C+\frac{j^*\epsilon}{n} \text{(by \ref{eq:inductivehypothesis})}= \\
& = DP_i +  \frac{j^*\epsilon}{n} \leq \\ 
& \leq DP_i +  \frac{(i-1)\epsilon}{n}
\end{align*}

Combining the above two inequalities, we obtain \ref{eq:suffices}. \textit{QED}

\end{proof}

The proof of Lemma~\ref{thrm:lemmalast} follows:

\begin{proof}
The proof is by contradiction. Suppose the optimal solution has a segment $[i,j]$ where $i \leq i_1 < i_2 \leq j$, and its optimal
$x$ value is $x^{*}$. Then consider splitting it into 5 intervals $[i, i_1-1]$, $[i_1, i_1]$, $[i_1+1, i_2-1]$, $[i_2, i_2]$, $[i_2+1, r]$.
If we let $x=x^{*}$ in the intervals not containing $i_1$ and $i_2$, their values are same as before.
Also, as  $|P_{i_1}-P_{i_2}|>2\sqrt{2C}$,
$ (P_{i_1}-x)^2 + (P_{i_2}-x)^2>2\sqrt{2C}^2 = 4C$. So by letting $x=P_{i_1}$ in $[i_1, i_1]$ and $x=P_{i_2}$ in $[i_2, i_2]$,
the total decreases by more than $4C$. This is more than the added penalty of having 4 more segments, a contradiction with the
optimality of the segmentation. ~\textit{QED}
\end{proof}

\end{document}